\documentclass{jnmp}

\usepackage{amsmath}
\usepackage[latin1]{inputenc} 
\def \bc {\begin{center}}
\def \ec {\end{center}}

\def \ba {\begin{array}}
\def \ea {\end{array}}

\def \bea {\begin{eqnarray}}
\def \eea {\end{eqnarray}}

\def \be {\begin{equation}}
\def \ee {\end{equation}}

\def \nn {\nonumber}
\def \um {\frac{1}{2}}

\JNMPnumberwithin{equation}{section}
\resetfootnoterule

\newtheorem{theorem}{Theorem}

\newtheorem*{acknowledgments}{Acknowledgments}
\theoremstyle{definition}
\newtheorem{definition}{Definition}

\begin{document}

\renewcommand{\evenhead}{M Calixto, V Aldaya, F F L\'{o}pez-Ruiz and E S\'{a}nchez-Sastre}
\renewcommand{\oddhead}{Coupling Nonlinear Sigma-Matter to Yang-Mills Fields: Symmetry Breaking Patterns}

\thispagestyle{empty}

\FirstPageHead{*}{*}{20**}{\pageref{firstpage}--\pageref{lastpage}}{Article}

\copyrightnote{200*}{M Calixto, V Aldaya, F F L\'{o}pez-Ruiz and E
S\'{a}nchez-Sastre}

\Name{Coupling Nonlinear Sigma-Matter to Yang-Mills Fields: Symmetry
Breaking Patterns}

\label{firstpage}

\Author{M. CALIXTO~$^{a,b}$, V. ALDAYA~$^b$, F. F. LOPEZ-RUIZ~$^b$ and E.
SANCHEZ-SASTRE~$^b$}

\Address{$^a$ Departamento de Matem\'{a}tica Aplicada y Estadística,
Universidad Polit\'{e}cnica de Cartagena, Spain \\ ~~E-mail:
Manuel.Calixto@upct.es\\[10pt]
$^b$ Instituto de Astrofísica de
Andalucía (CSIC), Granada, Spain
}

\Date{Received Month *, 200*; Accepted in Revised Form Month *, 200*}

\begin{abstract}
\noindent We extend the traditional formulation of Gauge Field Theory by
incorporating the (non-Abelian) gauge group parameters (traditionally
simple spectators) as new dynamical (nonlinear-sigma-model-type) fields.
These new fields interact with the usual Yang-Mills fields through a
generalized minimal coupling prescription, which resembles the so-called
Stueckelberg transformation \cite{Stueck}, but for the non-Abelian case.
Here we study the case of internal gauge symmetry groups, in particular,
unitary groups
$U(N)$. We show how to couple standard Yang-Mills Theory to
Nonlinear-Sigma Models on cosets of $U(N)$: complex projective, Grassman
and flag manifolds. These different couplings lead to distinct (chiral)
symmetry breaking patterns and \emph{Higgs-less} mass-generating
mechanisms for Yang-Mills fields.
\end{abstract}


\section{Introduction}
Although there has been many successful applications Non-Linear Sigma
Models (NLSM) in (Quantum Gauge) Field Theory, String Theory and
Statistical Mechanics, their basic role in Fundamental Physics is still
rather unexplored. Generally speaking, NLSM consists of a set of coupled
scalar fields
$\varphi^a(x^\mu), a=1,\dots,D,$ in a $d$-dimensional Minkowski spacetime
$M, \mu=0,1,2,\dots,d-1$, with the action
\be
S_\sigma=\lambda\int_M d^dx
g_{ab}(\varphi)\partial^\mu\varphi^a\partial_\mu\varphi^b,
\label{nlsmaction}\ee
where $\partial^\mu=\eta^{\mu\nu}\partial_\nu,
\partial_\nu=\partial/\partial x^\nu$, $\eta={\rm diag}(+,-,\dots,-)$ the Minkowski
metric and $\lambda$ a coupling constant. The field theory
(\ref{nlsmaction}) is called the NLSM with metric $g_{ab}(\varphi)$
(usually a positive-definite field-dependent matrix). The fields
$\varphi^a$ themselves can also be considered as the coordinates of an
internal Riemannian manifold $\Sigma$ with metric $g_{ab}$. In particular,
we shall consider the case in which
$\Sigma$ is a (semisimple) Lie group manifold $G$.

The relevance of NLSM in Quantum (Gauge) Field Theory originates from the
paramount importance of symmetry principles in fundamental physics. From
the String Theory point of view, the two-dimensional space $M$ represents
a string world sheet, whereas $g_{ab}$ is identified with the `truly'
spacetime metric representing the gravitational background where the
string propagates. In two dimensions we also have (infinite) conformal
symmetry and the possibility of adding new Wess-Zumino terms to our NLSM.

NLSM also provides a useful field-theoretical laboratory for studying some
two-dimensional, exactly solvable systems on a lattice, such as the Ising
model of the Heisemberg antiferromagnetism, in statistical mechanics. Some
particular $O(n)$-invariant two-dimensional NLSM are frequently used in
condensed matter physics in connection with antiferromagnetic spin chains
and the quantum Hall effect. Also, the effective Lagrangian for superfluid
He 3 is described by a NLSM. In four dimensions, pions and nucleons are
described by a (Skyrme) NLSM model, as solitonic solutions (`skyrmions').

We shall concentrate in the role that NLSM plays in the \emph{spontaneous
symmetry breaking} mechanism, which is crucial for phenomenological
applications of QFT like the Higgs-Kibble mechanism in the Standard Model
of Strong and Electro-Weak interactions, by means of which some vector
bosons acquire mass in a renormalizable way. According to the well known
Goldstone theorem (see e.g. \cite{Goldstone}), there are as many massless
(Nambu-Goldstone) particles as broken symmetry generators. If these
Nambu-Goldstone fields are scalars, their low energy effective action
often appears to be a NLSM. Usually, Goldstone bosons are eliminated from
the theory by gauge fixing.

Despite the undoubted success of the Standard Model in describing strong
and electro-weak interactions, a real (versus artificial) mechanism of
mass generation is still lacking. Needless to say that the discovery of a
Higgs boson (a quantum vibration of an abnormal Higgs vacuum) would be of
enormous importance; nevertheless, at present, no dynamical basis for the
Higgs mechanism exists and, as said, it is purely phenomenological. It is
true that there is actually nothing inherently unreasonable in the idea
that the state of minimum energy (the vacuum) may be one in which some
field quantity has a non-zero expectation value; in fact, many examples in
condensed-matter physics display this feature. Nevertheless, it remains
conjectural whether something similar actually happens in the weak
interaction case. Also, the ad hoc introduction of extra (Higgs) scalar
fields in the theory to provide mass to the vector bosons could be seen as
our modern equivalent of those earlier mechanical contrivances populating
the plenum (the ether), albeit very subtly. As in those days, new
perspectives are necessary to explain why it is really not indispensable
to look at things in this way at all.

One of the purposes of this paper is to provide a new formulation of gauge
theory in which the mass of gauge vector fields enters the theory in a
`natural' way without damaging gauge invariance. In this sense we shall
generalize the so-called Stueckelberg model for electrodynamics
\cite{Stueck} to account for a Higgs-less mass-generating mechanism for
gauge fields. In our new approach we shall nearly restrict the external
information to the symmetry group and, therefore, the group parameters,
described by Lagrangians of NLSM type, will acquire dynamical content as
`exotic' matter fields.\footnote{It is worth pointing out that the
incorporation of group parameters into some dynamical framework has
already been considered in other contexts, for example, in \cite{jackiw}.
There, conventional Eulerian fluid mechanics is extended to encompass the
possibility of describing a plasma state of quarks and gluons produced as
the result of high-energy collisions of heavy nuclei \cite{Ludlam} due to
the fact that such fluid may posses degrees of freedom indexed by group
variables.}. By the time being, we shall not enter into the possible
physical meaning of these $\sigma$-matter fields. Just to mention that,
when this idea is applied to the Weyl group (Poincar\'{e}+dilations), and
the corresponding gauge gravitational theory is developed, $\sigma$-fields
appear to be a natural source to account for some sort of \emph{dark
matter} intrinsically related to the gauge-group parameter associated with
scale transformations \cite{mpla}.

The underlying mathematical framework relies on the idea of {\it jet-gauge
group} \cite{rmp} introduced in Sec. \ref{jet-sec}. In Sec.
\ref{lf-jet-uti} we revise the Lagrangian formalism on jet-gauge groups
and generalize the well known Utiyama theorem \cite{Utiyama} which
provides a prescription to `minimally' couple Yang-Mills fields to
$\sigma$-matter fields. In Sec. \ref{chiral-break} we discuss several
\emph{chiral} gauge symmetry breaking patterns related to different mass
matrices. Sec. \ref{comments} is devoted to some comments on the
quantization of this model.

\section{Jet-Gauge Groups and Nonlinear $\sigma$-Fields}\label{jet-sec}

\begin{definition} \textbf{(Gauge group)} Let $G$ be a (matrix) Lie group
(the ``rigid'' group) and $M$ the Minkowski space-time (or any other
orientable space-time manifold). The gauge group $G(M)$ (``local'' or
current group) is the set of mappings
\begin{equation}
G(M)=\{g:M\to G,\, x\mapsto g(x)\}={\rm Map}(M,G)
\end{equation}
with point-wise multiplication $(gg')(x)=g(x)g'(x)$. The corresponding Lie
algebra
${\cal G}(M)$ is the tensor product
${\cal F}(M)\otimes {\cal G}=\{f^{a}X_a, \, a=1,\dots,{\rm dim}G\}$, where
${\cal F}(M)$ is the multiplicative algebra of $(C^\infty)$ differentiable
functions $f$ on $M$, and ${\cal G}$ is the Lie algebra of $G$ with
generators $X_a$. The commutation relations of this local algebra are
$[f\otimes X, h\otimes Y]=fh\otimes [X,Y]$ since, for internal symmetries, the ``rigid'' group
$G$ does not act on the space-time manifold $M$.
\end{definition}

We shall mainly consider special unitary groups $G=SU(N)$, the Lie algebra
of which ${\cal G}=su(N)=\langle X_a,\,a=1,\dots,N^2-1\rangle$ can be
expanded in terms of traceless hermitian matrices, $X_a$, whose
Lie-algebra commutators $[X_a,X_b]=C_{ab}^c X_c$ are given in terms of
totally antisymmetric structure constants $C_{ab}^c$. The generators $X_a$
can also be chosen to be orthogonal in the sense ${\rm
Tr}(X_aX_b)=\delta_{ab}$. A given group element $g\in G$ can be written in
terms of a (local) system of canonical coordinates
$\{\varphi^a, a=1,\dots,{\rm dim}(G)\}$ at the identity element as $g=e^{i\varphi^a X_a}$.
Thus, the composition group law $g''=g'g$ can also be locally written as:
\be
\varphi''^a=\varphi'^a+\varphi^a+\frac{1}{2}C_{bc}^a\varphi'^b\varphi^c+
\hbox{{\rm higher-order \ terms}}, \label{grouplaw}\ee
by using the Baker-Campbell-Hausdorff formula. Let us denote an element
$g(x)\in G(M)$ simply by its coordinates $\varphi^a(x)$ (the \textbf{exotic matter $\sigma$-fields}) .

\begin{definition} \textbf{(Jet prolongations)} Given a gauge group $G(M)$, we define the group $J^{1}(G(M))$ of the 1-jets of $G(M)$ as the quotient:
\[J^{1}(G(M))\equiv G(M)\times M/\sim ^{1}\]
where the equivalence relation $\sim ^{1}$ is defined as follows:
\[(\varphi,x)\sim ^{1} (\varphi',x')\Longleftrightarrow\left\{\ba{l} x=x',\\ \varphi(x)=\varphi'(x), \\
\partial_\mu{\varphi(x)}=\partial_\mu
\varphi'(x)\ea \right.\] for all
$(\varphi,x)$, $(\varphi',x')$ belonging to $G(M)\times M$. This definition may be easily extended from order $r=1$ to $r$-th order. A coordinate
system for $J^{1}(G(M))$ is $\{x^{\mu},\varphi^a,\varphi^a_{\mu}\}$.
\end{definition}
The formal definition of $J^{1}(G(M))$ is fully analogous to that of the
$(q^i,\,\dot{q}^j)$ phase-space in Lagrangian Mechanics, or
$(\psi^\alpha,\,\psi^\beta_\mu)$ in Lagrangian Field Theory, when one
desires to vary independently coordinates and velocities (momenta)
according to the modified Hamilton principle.

\begin{definition} \textbf{(Jet-gauge group)} We define the (infinite-dimensional) jet-gauge
group $G^1(M)$ as the set of mappings from M into $J^1(G(M))$:
\[ G^1(M)\equiv {\rm Map}(M, J^1(G(M))).\]
It is parametrized by the coordinate system
$\{\varphi^a(x),\varphi^a_\mu(x)\}$ and has the composition group law (\ref{grouplaw}), at each point $x\in M$, together with:
\be
\varphi''^a_\mu=\varphi'^a_\mu+\varphi^a_\mu+\frac{1}{2}C_{bc}^a\varphi'^b_\mu\varphi^c+\frac{1}{2}C_{bc}^a\varphi'^b\varphi^c_\mu+
{\rm higher \ order}.\label{grouplawjet}\ee
\end{definition}
In this formalism,
$\varphi^a_{\mu}$ are essentially the standard gauge vector potentials
(\textbf{Yang-Mills fields})
$A^{a}_{\mu}$ or connections, the actual relationship being:
 \be A^{a}_{\mu}\equiv \theta^{a}_b(\varphi)\varphi^b_{\mu} \,,\ee
where $\theta^{a}_b(\varphi)$ is the (non-constant) invertible matrix
defining the (left-) invariant canonical 1-form on the group
\be {\theta^L}^{a}=\theta^{a}_b(\varphi)d\varphi^b={\rm Tr}(i g^{-1}dg
X_a), \label{theta} \ee dual to the (left-invariant) vector fields \be
X^L_{a}=X^b_{a}(\varphi)\frac{\partial}{\partial
\varphi^b},\;\;X^b_{a}(\varphi)\equiv \frac{\partial
\varphi^{''b}(\varphi^{'},\varphi)}{\partial
\varphi^{a}}|_{\varphi=0,\varphi'=\varphi},\ee
that is: $\theta^{a}_bX^b_{c}=\delta ^a_c$. Writing then $A_\mu=i
g^{-1}g_\mu$, the group law (\ref{grouplawjet}) for Yang-Mills fields is
simply:
\[A''_\mu(x)=g^{-1}(x)A'_\mu(x)g(x)+A_\mu(x).\]
Note that $\varphi^a_{\mu}$ comprises all possible values of derivatives
of $\varphi^a$, but in general $\varphi^a_\mu\neq \partial_\mu \varphi^a$.
That is, not all Yang-Mills fields $A_\mu$ are ``pure gauge'',
$\theta_\mu=i g^{-1}\partial_\mu g$, except for the particular inmersion
(1-jet-extension) of the gauge group $G(M)$ into the jet-gauge group:
\be j^1: G(M)\rightarrow G^1(M), \,\,\varphi \mapsto
j^1(\varphi)=(\varphi^a,\partial_\mu\varphi^a). \ee

\section{Lagrangian Formalism on Jet-Gauge Groups: Generalized Utiyama
Theorem}\label{lf-jet-uti}

In the standard formulation of gauge theories, the well-known Minimal
Coupling Principle (or Utiyama theorem \cite{Utiyama}, see also
\cite{rmp}) for internal gauge symmetries establishes that if the action
of some matter fields
$\psi^{\alpha}$,
$\alpha=1,...,n$
\[S=\int \mathcal{L}_{{\rm m}}(\psi^{\alpha},\partial_{\mu}\psi^{\alpha})d^{4}x,\]
is invariant under a rigid internal Lie group $G$, then the modified
action
\[\widehat{S}=\int [\mathcal{L}_{{\rm m}}(\psi^{\alpha},D_{\mu}\psi^{\alpha})+\mathcal{L}_{0}(F^{a}_{\mu\nu})]d^{4}x\]
is invariant under the gauge group $G(M)$, where
\[D_{\mu}\psi^{\alpha}\equiv \partial_{\mu}\psi^{\alpha}-eA^{a}_{\mu}(X_a)^{\alpha}_{\beta}\psi^{\beta}\]
is usually known as the covariant derivative ($e$ is a coupling constant),
and
 \[F^{a}_{\mu\nu}\equiv\frac{1}{e}[D_\mu,\,D_\nu]^{a}=\partial_{\mu}A^{a}_{\nu}- \partial_{\nu}A^{a}_{\mu}+\frac{e}{2}C^{a}_{bc}(A^{b}_{\mu}A^{c}_{\nu}-A^{b}_{\nu}A^{c}_{\mu})\]
 is known as curvature of the connection $A^{a}_{\mu}$.

Here we shall treat the gauge group parameters $\varphi^a\in G(M)$ as
``exotic matter'' $\sigma$-fields, so that our configuration space is now
$J^1(G(M))$, with coordinates $\{x^{\mu},\varphi^a,A^{a}_{\mu}\}$, and
Lagrangians are accordingly functions
\[\mathcal{L}(x^{\mu},\varphi^a,A^{a}_{\mu};\partial_{\nu}\varphi^a,\partial_{\nu}A^{a}_{\mu}).\]
We shall proceed to formulate some sort of Minimal Coupling Principle on
$J^1(G(M))$:
\begin{theorem} \textbf{(Generalized Utiyama's Theorem)}
 If the action
\[
S_\sigma=\int \mathcal{L}_\sigma(\varphi^a,\partial_{\mu}\varphi^a)d^{4}x,
\]
of the ``exotic matter'' $\sigma$-fields
$\varphi^a,\,a=1,...,{\rm dim} G$ is invariant under the global (rigid) internal Lie group $G$,
 i.e.

\[
\delta^{\rm global}_{a}\mathcal{L}_\sigma(\varphi^b,\partial_\mu
\varphi^b)\equiv X^b_{a}\frac{\partial {\mathcal L}_\sigma}{\partial
\varphi^b}+ \frac{\partial X^b_{a}}{\partial \varphi^c}\partial_\mu
\varphi^c \frac{\partial {\mathcal L}_\sigma}{\partial (\partial_\mu
\varphi^b)}=0,
\]
then the modified action $S_{\rm tot}=\widetilde{S}_\sigma+S_0$, with
\be \widetilde{S}_\sigma\equiv \int \mathcal{L}_\sigma
(\varphi^a,\widetilde{D}_\mu\varphi^a)d^{4}x,\,\,\,S_0=\int\mathcal{L}_0(F^a_{\mu\nu})d^4x,
\ee
is invariant under the gauge (local) group $G(M)$, where
\[
 \widetilde{D}_{\mu}\varphi^a\equiv\partial_{\mu}\varphi^a-eA^{b}_{\mu}X^a_{b}
\]
is the ``covariant derivative'' for $\sigma$-fields.
\end{theorem}
 \begin{proof}
 As the local invariance of $S_0$ is already well-known in
the standard gauge theory,
 we shall focus on the local invariance of $\widetilde{S}_{\sigma}$. We must prove that
the new Lagrangian describing the gauge-group parameters as well as their
interaction with the gauge fields $A^{a}_\mu$ (according to the
prescription of Minimal Coupling, i.e. supposing that the group parameters
interact only with the gauge fields and not with their derivatives),

\[
\widetilde{\mathcal{L}}_{\sigma}(\varphi^a,\partial_\mu
\varphi^a,A^{a}_\mu)\equiv \mathcal{L}_{\sigma}(\varphi^a,\partial_\mu
\varphi^a-eA^{b}_\mu X^a_{b}),
\]

 is invariant under the gauge group $G(M)$ in the sense that

\[
\delta \widetilde{\mathcal{L}}_{\sigma}(\varphi^a,\partial_\mu \varphi^a,
A^{a}_\mu)\equiv f^{a}X^b_{a}\frac{\partial
\widetilde{\mathcal{L}}_{\sigma}} {\partial
\varphi^b}+\left(f^{a}\frac{\partial X^b_{a}}{\partial
\varphi^c}\partial_\mu \varphi^c+X^b_{a}\frac{\partial f^{a}}{\partial
x^\mu}\right)\frac{\partial \widetilde{\mathcal{L}}_{\sigma}} {\partial
(\partial_\mu \varphi^b)}\]\[+\left(gf^{b}C^a_{bc}A^{c}_\mu+\frac{\partial
f^{a}}{\partial x^\mu}\right) \frac{\partial
\widetilde{\mathcal{L}}_{\sigma}}{\partial A^{a}_\mu}=0,
\]
where $f^{a}$ denote gauge-algebra parameters.

 Let us consider the following change of variables:

\bea \phi^a&=&\varphi^a,\nn\\ \phi^a_\mu&=&\partial_\mu
\varphi^a-eA^{b}_\mu X^a_{b},\nn\\ B^{a}_\mu&=&A^{a}_\mu.\nn \eea

 Then, the partial derivatives related to the old variables can be
expressed in terms of the new ones:

\bea \frac{\partial\quad }{\partial \varphi^a}&=&\frac{\partial\quad
}{\partial \phi^a}- B^{c}_\mu \frac{\partial X^b_{c}}{\partial
\phi^a}\frac{\partial\quad\quad } {\partial (\phi^b_\mu)},\nn\\
\frac{\partial\quad\quad }{\partial (\partial_\mu \varphi^a)}&=&
\frac{\partial\quad\quad }{\partial (\phi^a_\mu)},\nn\\
\frac{\partial\quad }{\partial A^{a}_\mu}&=&\frac{\partial\quad }{\partial
B^{a}_\mu}- X^b_{a} \frac{\partial\quad\quad }{\partial (\phi^b_\mu)}.\nn
\eea

After this change of variables it is now straightforward to arrive at

\[\delta \widetilde{\mathcal{L}}_{\sigma}=f^{a}\delta^{\rm global}_{a}\mathcal{L}_{\sigma}(\phi^a,\phi^a_\mu)=0,\]
equality which follows from the hypothesis of invariance of the
$\sigma$-matter action under the global group.
\end{proof}

As a consequence, the new ``minimal coupling'' $\partial_\mu
\varphi^a\rightarrow \partial_\mu \varphi^a-eA^{b}_\mu X^a_{b}$ now occurs
in an affine manner. Indeed, the matrix $X^a_{b}$ is invertible, and
therefore the minimal coupling above is proportional to
\[ \theta^{a}_b[\partial_{\mu}\varphi^b-eX^b_{c}A^{c}_{\mu}]\equiv \theta^{La}_\mu-eA^{a}_\mu, \]
where $\theta^{La}$ is the canonical (left-)invariant 1-form in
(\ref{theta}). The new minimal coupling, when written in the form
$\theta^L-A$, strongly suggests the introduction of ``exotic matter'' of the
$\sigma$-model type:
\be
 \mathcal{L}_{\sigma}=\frac{\lambda^2}{2}{\rm Tr}_G(\theta^L_\mu{\theta^L}^\mu)=-\frac{\lambda^2}{2}
 {\rm Tr}_G(g^{-1}\partial_\mu gg^{-1}\partial^\mu g)\,. \label{traza}
\ee
This Lagrangian is $G$-invariant (left- and right-invariant, that is,
chiral) and the {\it new} minimal coupling gives rise to
\be \widetilde{\mathcal L}_{\sigma}=\frac{\lambda^2}{2}{\rm
Tr}_G[(\theta^L_\mu-eA_\mu)({\theta^L}^\mu-eA^\mu)],\label{mincoupsig} \ee
which is {\it gauge-invariant} even though it contains {\it mass terms}
$\frac{\lambda^2}{2} e^2{\rm Tr}_G[A_\mu A^\mu]$ for the $A_\mu$ fields, a piece which spoils
gauge invariance in the traditional framework of Yang-Mills theories.


\section{Chiral Symmetry Breaking Patterns}\label{chiral-break}

The virtue of a kinetic term like (\ref{traza}) is the two-side symmetry,
that is, \emph{chirality}. In fact, any function of $\theta^L$ is of
course left-invariant, but only a scalar sum on all the group indices
$a=1,\dots,{\rm dim}(G)$ can provide also right invariance. Therefore,
several chiral symmetry breaking patterns are possible by considering a
\emph{partial trace} $\sigma$-Lagrangian \be
 \mathcal{L}_{\sigma}^{(\lambda)}=\um {\rm Tr}_{G}^{(\lambda)}(\theta^L_\mu{\theta^L}^\mu)
 \equiv\um{\rm Tr}_{G}({\theta^L_\lambda}^\mu{\theta^L_\lambda}_\mu), \label{trazapar}
\ee
 where we have defined
\be
\theta^L_\lambda\equiv[\theta^L,\lambda]\ee
the `projection' of $\theta^L$ by the \emph{mass matrix}
$\lambda=i\lambda^aH_a$, with
 $\lambda^a\in\mathbb R$ and
$H_a$ the Lie algebra generators of the toral (Cartan, maximal Abelian) subgroup $H$ of $G$ (see later on this section for an example).
Defining
\[\Lambda\equiv g\lambda g^{-1},\,\,g\in G,\] (the adjoint action of $G$ on its Lie algebra) we have an alternative way of
writing (\ref{trazapar}) as \[\mathcal{L}_{\sigma}^{(\lambda)}=\um{\rm
Tr}_{G}(\partial_\mu\Lambda\partial^\mu\Lambda),\] which
 is singular due to the constraint ${\rm Tr}_{G}(\Lambda^2)={\rm
 Tr}_{G}(\lambda^2)=\lambda_a\lambda^a=$constant. Introducing Lagrange
 multipliers, the equations of motion read:
 \be
 \partial_\mu\partial^\mu\Lambda=-\frac{{\rm
 Tr}_{G}(\partial_\mu\Lambda\partial^\mu\Lambda)}{{\rm Tr}_{G}(\Lambda^2)}\Lambda,\ee
which describe a set of coupled Klein-Gordon-like fields $\phi^a={\rm
Tr}_{G}(\Lambda X_a)$ with variable mass $m^2={\rm
 Tr}_{G}(\partial_\mu\Lambda\partial^\mu\Lambda)/{\rm Tr}_{G}(\Lambda^2)$.

Let us explicitly consider the case of the unitary group
$G=U(N)$. We shall take, as the Lie algebra generators $X_a$, the step operators $X_{\alpha\beta}$ defined by the usual matrix elements:
\be
(X_{\alpha\beta})_{\gamma\rho}=\delta_{\alpha\gamma}\delta_{\beta\rho},\,\,\alpha,\beta,\gamma,\rho=1,\dots,N,
\ee
fulfilling the commutation relations:
\be
[X_{\alpha\beta},X_{\gamma\rho}]=\delta_{\gamma\beta}X_{\alpha\rho}-\delta_{\alpha\rho}X_{\gamma\beta},\ee
and the usual orthogonallity relations:
\be
{\rm
Tr}(X_{\alpha\beta}X_{\gamma\rho})=\delta_{\alpha\rho}\delta_{\gamma\beta}.
\ee
Note that the step generators $X_{\alpha\beta}$ are not hermitian but
$X_{\alpha\beta}^\dag=X_{\beta\alpha}$, where
$X^\dag$ denotes hermitian conjugate. This fact introduces some minor modifications
with respect to the general theory exposed before. For example, the
canonical left-invariant 1-form
$\theta^L$ can be written in this Lie-algebra basis as (we shall drop the
upper-script $L$ for convenience):
\be
\theta_\mu=\sum_{\alpha,\beta=1}^{N} \theta^{\alpha\beta}_\mu
X_{\alpha\beta},\ee
with $\theta^{\alpha\beta}=\bar{\theta}^{\beta\alpha}$ in order to make
$\theta^\dag=\theta$ (hermitian). The mass matrix $\lambda$ is now
\be
\lambda=i \sum_{\alpha=1}^{N} \lambda^\alpha X_{\alpha\alpha}, \ee
where the complex $i$ has been introduced in order to make the projected
1-form
\be\theta_\lambda=[\theta,\lambda]=-i\sum_{\alpha,\beta=1}^{N}
\theta^{\alpha\beta}(\lambda_\alpha-\lambda_\beta) X_{\alpha\beta}\ee
hermitian too.

When minimally coupled, like in (\ref{mincoupsig}), the partial trace
$\sigma$-Lagrangian (\ref{trazapar}) only assigns mass $m_{\alpha\beta}=e^2(\lambda_\alpha-\lambda_\beta)^2$ to those Yang-Mills
fields $A^{\alpha\beta}$ living on a certain coset
$G/G_\lambda$ of the group $G$, where $G_\lambda$ represents the `unbroken' chiral symmetry subgroup. Indeed,
the Lagrangian $\mathcal{L}_{\sigma}^{(\lambda)}$ is {\it left-}invariant
under the whole group $G$, but right-invariant under the unbroken subgroup
$G_\lambda$ only.

For $G=U(N)$ we can consider several symmetry breaking patterns according
to distinct mass matrix $\lambda$ choices:
\begin{enumerate}
\item  For the case $\lambda_\alpha\not=\lambda_\beta,\forall \alpha,\beta=1\dots,N$ the unbroken symmetry is
$G_\lambda=U(1)^{N}$, so that we give mass to all of $N(N-1)/2$ charged (complex) Yang-Mills
fields $A^{\alpha\beta}, \alpha>\beta$ (the analogue of $W_\pm$ in $U(2)$
invariant electro-weak model \cite{Goldstone}) living on the \emph{flag
manifold} (coset)
$\mathbb F_N=G/G_\lambda=U(N)/U(1)^{N}$ . The neutral (not charged) vector
bosons $A^{\alpha\alpha}$ remain massless.
\item For $\lambda_\alpha=\lambda_\beta, \forall \alpha,\beta=2,\dots,N$
the unbroken symmetry is
$G_\lambda=U(N-1)\times U(1)$, so that we have $N-1$ massive charged Yang-Mills fields $A^{1\alpha}, \alpha>1$ living
on the \emph{complex projective} space
$\mathbb CP^{N-1}=U(N)/U(N-1)\times U(1)$, in addition to $(N-1)(N-2)/2$ massless charged vector bosons
$A^{\alpha\beta}, \alpha\not=\beta\not=1$ and $N$ massless neutral vector bosons
$A^{\alpha\alpha}$.
\item For other choices like:
\[\lambda_1=\lambda_2=\dots=\lambda_{N_1}\not=\lambda_{N_1+1}\not=\dots\not=\lambda_{N-N_2}=\dots=\lambda_{N}\]
the unbroken symmetry group is $G_\lambda=U(N_1)\times U(N_2)\times U(1)$
giving $N_1(N_1-1)/2+N_2(N_2-1)/2$ massless charged vector bosons, $N$
massless neutral vector bosons and massive charged vector bosons
corresponding to the \emph{complex Grasmannian}
$\mathbb CG(N_1,N_2)=U(N)/U(N_1)\times U(N_2)\times U(1)$ (see \cite{jgp} for suitable coordinate
systems on these coset spaces).
\end{enumerate}
Note that this `partial trace' mechanism always keeps the $N$ neutral
vector bosons $A^{\alpha\alpha}$ massless. However, we could always supply
mass to the neutral vector boson $Z_0$, related to the central generator
$H_0=\sum_{\alpha=1}^N X_{\alpha\alpha}$ (which commutes with everything), without spoiling
the previous mechanism, using the conventional Stueckelberg model for the
Abelian case $G=U(1)$.

These whole scheme agrees with nature, where we find just one intermediate
massive neutral vector boson $Z_0$ (inside weak currents); the rest of
intermediate neutral vector bosons (photon and gluons) remain massless.

\section{Comments and Outlook}\label{comments}
The fact that both the gauge functions $\varphi$ and the vector potentials
$A$ themselves may be considered as parameters of a group,
$G^1(M)$, which constitutes the basic symmetry group of the theory (in the
sense that the corresponding Noether invariants parametrizes the solution
manifold), permits to face the quantum theory under the perspective of a
non-perturbative group-theoretical framework (according to the scheme
outlined in Ref. \cite{jpa}) where questions such as renormalizability,
finiteness, unitarity, etc., are much better addressed. The Hilbert space
of our theory will be the carrier space of unitary irreducible
representations of a centrally extended infinite-dimensional Lie group
$\tilde G$, incorporating $G^1(M)$ and the phase space of our theory.

Let us make a brief discussion of the physical field degrees of freedom of
our theory. This analysis of the dynamical content of the theory can be
achieved without the need of writing down the explicit expression of the
(linearized) field equations of motion. Instead, we shall resort again to
a group-representation viewpoint at the Lie algebra level (see \cite{jpa}
for more precise details on a Group Approach to Quantization of Yang-Mills
theories). In fact, for pure, massless,
$SU(N)$-Yang-Mills theory we can fix the (Weyl) gauge and set the temporal
part
$A_0^a=0, a=1,\dots,N^2-1$. The equal-time  Lie algebra commutators between non-Abelian vector potentials
$A_j^a, j=1,2,3; a=1,\dots,N^2-1$, electric field
$E_j^a$ and gauge-group generators $\varphi^a$ (in natural $\hbar=1=c$ unities) turn out to be
(see e.g. the Reference \cite{Jackiw2}):
\bea \left[A_j^a(x),E_k^b(y)\right] &=&
i\delta_{jk}\delta^{ab}\delta(x-y),\nn
\\
 \left[{\vec{E}}^a(x),\varphi^b(y)\right] &=& -iC^{ab}_c {\vec{E}}^c(x)\delta(x-y), \nn\\
\left[\vec{A}^a(x),\varphi^b(y)\right] &=& -iC^{ab}_c
\vec{A}^c(x)\delta(x-y)
-\frac{i}{e}\delta^{ab}\vec{\nabla}_x\delta(x-y),\nn
\\ \left[\varphi^a(x),\varphi^b(y)\right] &=& -iC^{ab}_c
\varphi^c(x)\delta(x-y).\label{ym-com} \eea
From the first commutator we see that
$A_j^a$ and
$E_j^a$ are conjugated variables, so that we have in principle three field
degrees of freedom for each ``colour'' index
$a=1,\dots,N^2-1$, that is,  $f=3(N^2-1)$ original field degrees of freedom.
All $\sigma$-fields $\varphi^a, a=1,\dots,N^2-1,$ do not have dynamics
this time, so that we can impose all of them as constraints
$\varphi^a(x)\Psi=0$ (the \emph{Gauss law}) on wave functionals $\Psi$ in
the corresponding quantum field theory. This operation takes away
$c=N^2-1$ field degrees of freedom out of the original $f$, leaving
$f'=f-c=2\times(N^2-1)$. These field degrees of freedom correspond to $(N^2-1)$
massless vector bosons (remember that transversal fields have two
polarizations only).

When we give dynamics to some of the $\sigma$-fields $\varphi^a$ through a
partial-trace $\sigma$-Lagrangian like (\ref{trazapar}), and perform
minimal coupling $\theta_\mu\to\theta_\mu-eA_\mu$, we introduce new
conjugated variables given by the new Lie-algebra commutators:
\be \left[{A}^a_0(x),\varphi^b(y)\right] = -iC^{ab}_c
{A}^c_0(x)\delta(x-y) -\frac{i}{e}C^{ab}_c\lambda^c\delta(x-y),
\label{euclidean}\ee
where (in the hope that no confusion arises) we mean here by ${A}^a_0(x)$
the generator of translations in the temporal component of the vector
potential. It should be stressed that the central term proportional to
$\lambda^c$ in the previous commutator can also be considered as
associated with some sort of ``symmetry breaking'' in the sense that it
can be hidden into a redefinition, $A^c_0\rightarrow
A^c_0+\frac{\lambda^c}{e}$, of $A^c_0$, which now acquires a non-zero
vacuum expectation value proportional to the mass $\lambda^c$, that is:
\[\langle 0|A^c_0|0\rangle=0\longrightarrow \langle 0|A^c_0|0\rangle=-\frac{\lambda^c}{e}.\]
This is one of the differences between the vacua of the massless and the
massive theory. Moreover, $\sigma$-fields could also acquire non-zero
vacuum expectation values, $\langle 0|\varphi^c|0\rangle=\omega^c$, which
could be mimicked by new central terms in the last commutator of
(\ref{ym-com}). See Ref. \cite{jpa} for the physical consequences of this
particular case.

Let us proceed by counting the new physical field degrees of freedom of
the massive theory. We shall restrict ourselves to
$G=SU(2)$ (i.e., $N=2$), for the sake of simplicity, and take $\lambda=i\lambda^3 T_3$
(the ``isospin'' charge). Let us also use the Cartan basis $\langle
T_\pm=T_1\pm iT_2, T_0=T_3\rangle$, with commutation relations:
\[
[T_\pm,T_0]=\mp T_\pm, \;\;[T_+,T_-]=2T_0.\]
The commutation relations (\ref{euclidean}) say that the temporal part
$W^±_0\equiv A^1_0\pm i A^2_0$ (we adopt the usual notation in the Standard Model of electro-weak
interactions for charged weak vector bosons) are conjugated fields of
$\varphi^{\mp}\equiv\varphi^1\mp \varphi^2$, since they give central terms
proportional to the mass matrix element $\lambda^3$. On the contrary, the
temporal part $B_0\equiv A^3_0$ and the $\sigma$-field
$\varphi^0\equiv \varphi^3$ remain without dynamics. Thus, in addition to the
original $f=3(N^2-1)=9$ field degrees of freedom connected to the spatial
part $\vec{A}^a, a=1,2,3$, we have two additional field degrees of freedom
attached to the temporal part
$W^±_0$, which results in $\tilde f=f+2=11$ field degrees of freedom. If we wished to be
consistent with the massless case, we should gauge fix $\sigma$-fields to
zero as constraints $\varphi^a(x)\Psi=0$ in the quantum theory. This
operation would take away $c=(N^2-1)=3$ field degrees of freedom out of
the original $\tilde f=11$, leaving $\tilde f'=\tilde f-c=8=2\times
1+3\times 2$. These field degrees of freedom correspond to $1$ massless
vector boson $B$ (two polarizations) plus $2$ massive vector bosons
$W^\pm$. We could say that the dynamics of the $\sigma$-fields
$\varphi^\pm$ has been transferred to the vector potentials $W^\pm$ (the longitudinal part) to conform
massive vector bosons. Hence, we do not need an extra (Higgs) field to
give mass to vector bosons but it is the gauge group itself which acquires
dynamics and transfers it to Yang-Mills fields.

A deeper (Lagrangian) analysis of the particular case of electro-weak
gauge group $SU(2)\otimes U(1)$, in the framework of the Standard Model,
is in preparation \cite{preparation}. Also, a proper Group Approach to
Quantization of this theory, clarifying the vacuum and including
interaction with fermions and comparisons with the Standard Model, as well
as a deeper discussion on the physical status of
$\sigma$-fields, is being investigated by the authors.

\begin{acknowledgments}
Work partially supported by the MCYT and Fundaci\'{o}n S\'{e}neca under
projects FIS2005-05736-C03-01 and 0310/PI/05
\end{acknowledgments}

\label{lastpage}

\end{document}